\documentclass[11pt]{article}
\usepackage{mathrsfs}
\usepackage[centertags]{amsmath}
\usepackage{amsfonts}
\usepackage{amssymb}
\usepackage{amsthm}
\usepackage{cases}
\usepackage{indentfirst}
\usepackage{epsfig}
\usepackage{graphicx}
\usepackage{color}
\usepackage{CJK}
\usepackage[numbers,sort&compress]{natbib}
\usepackage{subfigure} 
\usepackage[noend]{algpseudocode}
\usepackage{algorithmicx,algorithm}

\date{}

\newtheorem{theorem}{Theorem}[section]
\newtheorem{corollary}{Corollary}[section]
\newtheorem{lemma}{Lemma}[section]

\newtheorem{remark}{Remark}[section]
\newtheorem{definition}{Definition}[section]
\numberwithin{equation}{section}

\allowdisplaybreaks[4]

\begin{document}
\title{\textbf{The high-order block RIP for non-convex block-sparse compressed sensing}}
\author{$^{a}$Jianwen Huang \qquad$^{b,c}$Xinling Liu\qquad $^{b,c}$Jinyao Hou \qquad $^{c}$Jianjun Wang \thanks{Corresponding author, E-mail: wjjmath@gmail.com, wjj@swu.edu.cn(J.J. Wang), E-mail: hjw1303987297@126.com (J. Huang)}\\
{\small $^{a}$School of Mathematics and Statistics, Tianshui Normal University, Tianshui, 741001, China}\\
{\small $^{b}$School of Mathematics and Statistics, Southwest University, Chongqing, 400715, China}\\
{\small $^c$College of Artificial Intelligence, Southwest
University, Chongqing, 400715, China} }

\maketitle
\begin{quote}
{\bf Abstract.}~~This paper concentrates on the recovery of block-sparse signals, which is not only sparse but also nonzero elements are arrayed into some blocks (clusters) rather than being arbitrary distributed all over the vector, from linear measurements. We establish high-order sufficient conditions based on block RIP to ensure the exact recovery of every block $s$-sparse signal in the noiseless case via mixed $l_2/l_p$ minimization method, and the stable and robust recovery in the case that signals are not accurately block-sparse in the presence of noise. Additionally, a lower bound on necessary number of random Gaussian measurements is gained for the condition to be true with overwhelming probability. Furthermore, the numerical experiments conducted demonstrate the performance of the proposed algorithm.

{\bf Keywords.}~~Compressed sensing; block restricted isometry property; block sparsity; mixed $l_2/l_p$ minimization.
\end{quote}

\section{Introduction}
\label{sec.1}

Block-sparse signal recovery (BSR) appears in some fields of sparse modelling and machine learning, including color imaging \cite{Majumdar and Ward 2010}, equalization of sparse communication channels \cite{Parvaresh et al 2008}, multi-response linear regression \cite{Cotter and Rao 2002} and imagine annotation \cite{Huang J et al 2002} and so forth. Essentially, the important problem in BSR is how to reconstruct a block-sparse or approximately block-sparse signal from a linear system. Commonly, one thinks over the below model:
\begin{align}
\notag y=\Phi x+e,
\end{align}
where $y\in\mathbb{R}^n$ is the observation measurement, $\Phi\in\mathbb{R}^{n\times N}$ is a known measurement matrix (or sensing matrix) with $n<N$, and $e\in\mathbb{R}^n$ is a vector of measurement errors. Generally, the conventional compressed sensing (CS) simply thinks out the sparsity of the signal to be recovered, however it doesn't consider any additional structure, i.e., non-zero elements appear in blocks (or clusters) rather than being arbitrarily spread all over the vector. We call these signals as the block-sparse signals. In order to define block sparsity, we need to give several additional notations. Suppose that the signal $x$ over the block index set $\mathcal{I}=\{d_1,d_2,\cdots,d_M\}$, then we can describe the signal $x$ as
\begin{align}\label{eq.1}
x=[\underbrace{x_1,\cdots,x_{d_1}}_{x[1]},\underbrace{x_{d_1+1},\cdots,x_{d_1+d_2}}_{x[2]},
\cdots,\underbrace{x_{N-d_M+1},\cdots,x_N}_{x[M]}]^{\top},
\end{align}
where $x[i]$ represents the $i$th block of $x$ and $N=\sum^M_{i=1}d_i$. We call a vector $x$ as block $s$-sparse over $\mathcal{I}=\{d_1,d_2,\cdots,d_M\}$ if $x[i]$ is non-zero for no more than $s$ indices $i$ \cite{Eldar and Mishali 2009}. In fact, if the block-sparse structure of signal is neglected, the conventional CS doesn't efficiently treat such structured signal. To reconstruct block-sparse signal, researchers \cite{Eldar and Mishali 2009} \cite{Eldar et al 2010} proposed the following mixed $l_2/l_1$-minimization:
\begin{align}\label{eq.2}
\hat{x}=\arg\min_{\tilde{x}\in\mathbb{R}^N}\|\tilde{x}\|_{2,\mathcal{I}}~\mbox{s.t.}~y-\Phi \tilde{x}\in\mathcal{B},
\end{align}
where $\|x\|_{2,\mathcal{I}}=\sum^M_{i=1}\|x[i]\|_2$ is the mixed $l_2/l_1$ norm of a vector $x$. The set $\mathcal{B}$ stands for some noise structure,
\begin{align}\label{eq.9}
\mathcal{B}^{l_2}(\rho):=\{e:~\|e\|_2\leq\rho\}
\end{align}
and
\begin{align}\label{eq.10}
\mathcal{B}^{DS}(\rho):=\{e:~\|\Phi^{\top}e\|_{\infty}\leq\rho\},
\end{align}
where $\Phi^{\top}$ represents the conjugate transpose of the matrix $\Phi$. (\ref{eq.2}) is a convex optimization issue and could be converted into a second-order cone program, so can be solved efficiently.

To study the theoretical performance of mixed $l_2/l_1$-minimization, Eldar and Mishali \cite{Eldar and Mishali 2009} proposed the definition of block restricted isometry property (block RIP).
\begin{definition}\label{def.1}(Block RIP \cite{Eldar and Mishali 2009})
Given a matrix $\Phi\in\mathbb{R}^{n\times N}$, for every block $s$-sparse $x\in\mathbb{R}^N$ over $\mathcal{I}=\{d_1,d_2,\cdots,d_M\}$, there is a positive number $\delta\in(0,1)$, if
\begin{align}\label{eq.3}
(1-\delta)\|x\|^2_2\leq\|\Phi x\|^2_2\leq(1+\delta)\|x\|^2_2,
\end{align}
then the matrix $\Phi$ obeys the $s$-order block RIP over $\mathcal{I}$. Define the block RIP constant (RIC) $\delta_{s|\mathcal{I}}$ as the smallest positive constant $\delta$ such that (\ref{eq.3}) holds for all $x\in\mathbb{R}^N$ that are block $s$-sparse.
\end{definition}
For the reminder of this paper, for simplicity, $\delta_s$ represents the block RIP constant $\delta_{s|\mathcal{I}}$. Eldar and Mishali \cite{Eldar and Mishali 2009} showed that the mixed $l_2/l_1$-minimization method can exactly reconstruct any block $s$-sparse signal when the measurement matrix $\Phi$ fulfills the block RIP with $\delta_{2s}<0.414$. Later, Lin and Li \cite{Lin and Li 2013} improved the sufficient condition on $\delta_{2s}$ to $0.4931$, and builded condition $\delta_s<0.307$ for accurate recovery. In 2019, the conclusions of literature \cite{Li and Chen 2019} and \cite{Huang JW et al 2019} together present a complete characterization to the block RIP condition on $\delta_{ts}$ that the mixed $l_2/l_1$ minimization method guarantees the block-sparse signal recovery in the field of block-sparse compressed sensing.

Recently, a lot of researchers \cite{Chartrand and Staneva 2008} \cite{Shen and Li 2012} \cite{Lai MJ et al 2013} have revealed that $l_p~(0<p<1)$ minimization not only constantly needs less constrained the RIP requirements, but also could ensure exact recovery for smaller $p$ compared with the $l_1$ minimization. In the present paper, we are interested in investigating the stable reconstruction of block-sparse signals by the mixed $l_2/l_p~(0<p<1)$ minimization as follows:
\begin{align}\label{eq.4}
\hat{x}=\arg\min_{\tilde{x}\in\mathbb{R}^N}\|\tilde{x}\|^p_{2,p}~\mbox{s.t.}~y-\Phi \tilde{x}\in\mathcal{B},
\end{align}
where $\|x\|^p_{2,p}=\sum^M_{i=1}\|x[i]\|^p_2$. Simulation experiments \cite{Majumdar and Ward 2010} \cite{Wang Y et al 2013} indicated that fewer linear measurements are needed for accurate reconstruction when $0<p<1$ than when $p=1$. More related work can be found in literature \cite{Wang Y et al 2014} \cite{Wen JM et al 2019} \cite{Gao Y et al 2017} \cite{Wang JJ et al 2019} \cite{Ge and Chen 2018} \cite{Li and Wen 2019} \cite{Wang WD et al 2017}.

In this paper, we further investigate the high-order block RIP conditions for the exact and stable reconstruction for (approximate) block-sparse signals by mixed $l_2/l_p$ minimization. The crux is extend sparse representation of an $l_p$-polytope (Lemma 2.2 \cite{Zhang and Li 2019}) to the block scenario. With this technique, we obtain a sufficient condition on RIC $\delta_{ts}$ that guarantees the exact and stable reconstruction of approximate block-sparse signals via mixed $l_2/l_p$ minimization, and establish error bounds between the solution to (\ref{eq.4}) and the signal $x$ to be recovered. Obviously, when $x$ is accurately block-sparse and $\mathcal{B}=\{0\}$ (i.e., $y=\Phi x$), we will derive the accurate reconstruction condition. Particularly, we will determine how many random Gaussian measurements suffice for the condition to hold with high probability.

The remainder of the paper is constructed as follows. In Section \ref{sec.2}, we will provide some notations and a few lemmas. In Section \ref{sec.3}, we will present the main results, and the associating proofs are given in Section \ref{sec.4}. In Section \ref{sec.5}, a series of numerical experiments are presented to support our theoretical results. Lastly, the conclusion is drawn in Section \ref{sec.6}.

\section{Preliminaries}
\label{sec.2}

Throughout this article, we use the below notations unless special mentioning. For a subset $T$ in $\mathbb{R}^M$, $T^c$ denotes the complement of $T$ in $\mathbb{R}^M$. For any vector $x\in\mathbb{R}^N$, $x_T$ represents a vector which is equal to $x$ on block indices $T$ and displaces other blocks with zero. Denote $T_0$ by block indices of the $s$ largest block in $l_2$ norm of the vector $x$, i.e., $\|x[i]\|_2\geq\|x[j]\|_2$ for any $i\in T_0$ and $j\in T^c_0$. We represent $x_{\max(s)}$ as $x$ with all but the largest $s$ blocks in $l_2$ norm set to zero. Henceforth, we invariably choose that $h=\hat{x}-x_{\max(s)}$, where $\hat{x}$ is the minimizer of (\ref{eq.4}).

In order to prove our main results, it is necessary to present the below lemma which is a crucial technical tool. Factually, we extend sparse expression of an $l_p$-polytope proposed by \cite{Zhang and Li 2019} to the block context.

\begin{lemma}\label{lem.1}

For a positive integer $s$, a positive number $\alpha$ and given $p\in(0,1]$, define the block $l_p$-polytope $T(\alpha,s,p)\in\mathbb{R}^N$ by
\begin{align}\notag
T(\alpha,s,p)=\{x\in\mathbb{R}^N:~\|x\|^p_{2,p}\leq s\alpha^p,~\|x\|_{2,\infty}\leq \alpha\}.
\end{align}
Then any $x\in T(\alpha,s,p)$ can be expressed as the convex combination of block $s$-sparse vectors, i.e.,
\begin{align}\notag
x=\sum_i\lambda_iu_i.
\end{align}
Here $\lambda_i>0$ and $\sum_i\lambda_i=1$ and $\|u_i\|_{2,0}\leq s$. In addition,
\begin{align}\label{eq.5}
\sum_i\lambda_i\|u_i\|^2_{2,2}\leq\alpha^p\|x\|^{2-p}_{2,2-p}.
\end{align}
\end{lemma}

\begin{proof}
We can prove the assertion holds by induction. If $x$ is block $s$-sparse, we can set $u_1=x$ and $\lambda_1=1$, then $\|u_1\|^2_{2,2}=\|x\|^2_{2,2}\leq\alpha^p\|x\|^{2-p}_{2,2-p}$. Suppose that assertion holds for all block $(l-1)$-sparse vectors $x$ ($l-1\geq s$). Then for any block $l$-sparse vectors $x$ such that $\|x\|^p_{2,p}\leq s\alpha^p$ and $\|x\|_{2,\infty}\leq \alpha$, without loss of generality suppose that $x$ is not block $(l-1)$-sparse (otherwise the statement is naturally true by assumption of $l-1$). Besides, $x$ can be represented as $x=\sum^l_{i=1}c_iE_i$, where $c_1\geq c_2\geq\cdots\geq c_l>0$, $c_1$ is equal to the largest $\|x[i]\|_2$ for every $i\in\{1,2,\cdots,M\}$, $c_2$ is equal to the next largest $\|x[i]\|_2$, etc. Here $E_i$ denotes a unit vector in $\mathbb{R}^N$, which is equal to $x/c_i$ on the $i$th largest block of $x$ and zero other places. Set $c_0=\alpha$, and fix $a=(a_1,a_2,\cdots,a_l)\in\mathbb{R}^l_+$, where $a_i=c_i^{p-1},~i=0,1,\cdots,l$. Then we have $\sum^l_{i=1}a_ic_i\leq s\alpha^p$ and $\alpha^p=a_0c_0\geq a_1c_1\geq\cdots\geq a_lc_l.$

Denote the set
\begin{align}\label{eq.6}
\Gamma=\{1\leq j\leq l-1:~\sum^l_{i=j}a_ic_i\leq (l-j)a_{j-1}c_{j-1}\}.
\end{align}
It is easy to see that $1\in\Gamma$, hence $\Gamma$ is not empty. Then we note that $j=\max\Gamma$, which implies
\begin{align}\notag
\sum^l_{i=j}a_ic_i&\leq (l-j)a_{j-1}c_{j-1},\\
\sum^l_{i=j+1}a_ic_i&> (l-j-1)a_{j}c_{j}.
\end{align}
It follows that
\begin{align}\notag
(l-j)a_{j}c_{j}<\sum^l_{i=j}a_ic_i\leq(l-j)a_{j-1}c_{j-1}.
\end{align}
Set
\begin{align}\notag
y_w&=\sum^{j-1}_{i=1}c_iE_i+\frac{\sum^l_{i=j}a_ic_i}{l-j}\sum^l_{i=j,i\neq w}a^{-1}_iE_i,\\
\notag \xi_w&=1-\frac{l-j}{\sum^l_{i=j}a_ic_i}a_wc_w,
\end{align}
where $w=j,j+1,\cdots,l$. Then by simple calculations, we obtain $\sum^l_{w=j}\xi_w=1$ and $x=\sum^l_{w=j}\xi_wy_w$, where $y_w$ is block $(l-1)$-sparse for all $w=j,j+1,\cdots,l$. Finally, since $y_w$ is block $(l-1)$-sparse, under the induction assumption, we have $y_w=\sum_i\mu_{w,i}u_{w,i}$, where $u_{w,i}$ is block $s$-sparse, and $\mu_{w,i}\in[0,1]$, $\sum_i\mu_{w,i}=1$. Hence, $x=\sum_i\sum^l_{i=j}mu_{w,i}u_{w,i}$, which implies that statement is true for $l$.
\end{proof}

We will utilize the below lemma in the process of proving the main conclusions, which is a useful important inequality.
\begin{lemma}\label{lem.2}(Lemma 5.3 \cite{Cai and Zhang 2013})
Suppose that $M\geq s$, $a_1\geq a_2\geq\cdots\geq a_M\geq 0$, $\sum^s_{i=1}a_i\geq\sum^M_{i=s+1}a_i$, then for all $\alpha\geq 1$,
\begin{align}\notag
\sum^M_{j=s+1}a^{\alpha}_j\leq\sum^s_{i=1}a^{\alpha}_i.
\end{align}
More generally, assume that $a_1\geq a_2\geq\cdots\geq a_M\geq 0$, $\lambda\geq0$ and $\sum^s_{i=1}a_i+\lambda\geq\sum^M_{i=s+1}a_i$, then for all $\alpha\geq 1$,
\begin{align}\notag
\sum^M_{j=s+1}a^{\alpha}_j\leq s\left(\sqrt[\alpha]{\frac{\sum^s_{i=1}a^{\alpha}_i}{s}}
+\frac{\lambda}{s}\right)^{\alpha}.
\end{align}
\end{lemma}

In view of the definition of $h$, $\hat{x}$ and $x_{\max(s)}$, we get the below lemma.
\begin{lemma}\label{lem.3}
Recall that $h=\hat{x}-x_{\max(s)}$, where $\hat{x}$ is the solution to (\ref{eq.4}). It holds that
\begin{align}\notag
\|h_{-\max(s)}\|^p_{2,p}\leq\|h_{\max(s)}\|^p_{2,p}.
\end{align}
\end{lemma}
\begin{proof}
Assume that $T_0$ is the block index set over $s$ blocks with largest $l_2$ norm of the vector $x$. Therefore, $x_{T_0}=x_{\max(s)}$. By applying the minimality of the solution $\hat{x}$ and the reverse triangular inequality of $\|\cdot\|^p_{2,p}$, we get
\begin{align}\notag
\|x_{T_0}\|^p_{2,p}\geq\|\hat{x}\|^p_{2,p}&=\|x_{T_0}+h_{T_0}\|^p_{2,p}+\|h_{T^c_0}\|^p_{2,p}\\
\notag&\geq\|x_{T_0}\|^p_{2,p}-\|h_{T_0}\|^p_{2,p}+\|h_{T^c_0}\|^p_{2,p},
\end{align}
which implies
\begin{align}\notag
\|h_{T^c_0}\|^p_{2,p}\leq\|h_{T_0}\|^p_{2,p}.
\end{align}
Note that $\|h_{-\max(s)}\|^p_{2,p}\leq\|h_{T^c_0}\|^p_{2,p}$ and $\|h_{T_0}\|^p_{2,p}\leq\|h_{\max(s)}\|^p_{2,p}$. Combining with the above inequalities, the desired result can be derived.
\end{proof}

\begin{lemma}\label{lem.4}(Lemma 5.1 \cite{Baraniuk et al 2008})
Let $\Phi\in\mathbb{R}^{n\times N}$ be a random matrix whose entries obey one of the distributions given by (\ref{eq.33}) and that fulfills the concentration inequality
\begin{align}\label{eq.34}
\mathbb{P}(|\|\Phi x\|^2_2-\|x\|^2_2|\geq\epsilon\|x\|^2_2)\leq2e^{-nc_0(\epsilon)},~\epsilon\in(0,1),
\end{align}
where the probability is taken over all $n\times N$ matrices $\Phi$ and $c_0(\epsilon)$ is a constant relying merely on $\epsilon$ and such that for all $\epsilon\in(0,1)$, $c_0(\epsilon)>0$. Suppose that $1\leq s\leq n$. Then, for any $\delta\in(0,1)$, we have
\begin{align}\label{eq.35}
(1-\delta)\|x\|^2_2\leq\|\Phi x\|^2_2\leq(1+\delta)\|x\|^2_2
\end{align}
for all $s$-sparse vectors $x\in\mathbb{R}^N$ with probability
\begin{align}\label{eq.36}
\geq 1-2\left(\frac{12}{\delta}\right)^se^{-c_0\left(\frac{\delta}{2}\right)n}.
\end{align}
\end{lemma}

\section{Main Results}
\label{sec.3}

Based on the knowledge prepared above, we present the main results in this part-a high-order block RIP condition for the robust reconstruction of arbitrary signals with block structure via mixed $l_2/l_p$ minimization. In the case that the signal to be recovered is block-sparse, the condition can respectively guarantee the accurate construction and stable recovery in the noise-free case and in the noise situation. When the original signal $x$ is not block-sparse and the linear measurement is corrupted by noise, the below result presents a sufficient condition for recovery of structured signals.
\begin{theorem}\label{the.1}
Let $y=\Phi x+e$ be noisy measurements of a signal $x\in\mathbb{R}^N$ with $y,~e\in\mathbb{R}^n$, $\Phi\in\mathbb{R}^{n\times N}~(n<N)$ and $\|e\|_2\leq\rho$. Assume that $\mathcal{B}=\mathcal{B}^{l_2}(\varepsilon)$ with $\rho +\sigma(\Phi)\|x_{-\max(s)}\|_2\leq\varepsilon$ in (\ref{eq.4}). If $\Phi$ satisfies the block RIP with
\begin{align}\label{eq.7}
\delta_{ts}<\frac{\mu}{\frac{2-p}{t-1}-\mu}:=\phi(t,p)
\end{align}
for some $1<t\leq 2$, where $\mu\in[(\sqrt{1+2p-p^2}-1)/p,(1-(t-\sqrt{t^2-t})p)/(t-1)]$ is the sole positive solution of the equation
\begin{align}\label{eq.8}
g(\mu,p)=\frac{p}{2}\mu^{\frac{2}{p}}+\mu-\frac{2-p}{2(t-1)}.
\end{align}
Then the solution $\hat{x}^{l_2}$ to (\ref{eq.4}) fulfills
\begin{align}\label{eq.12}
\|\hat{x}^{l_2}-x\|_2\leq C_1(\varepsilon+\rho)+C_2\|x_{-\max(s)}\|_2,
\end{align}
where
\begin{align}\notag
C_1&=\sqrt{2}\left(\frac{\phi(t,p)}{\phi(t,p)-\delta_{ts}}\frac{(2-p)(1-(t-1)\mu)}{2-p-(t-1)\mu}
\sqrt{1+\delta_{ts}}+\phi(t,p)\sqrt{\frac{1-p}{\phi(t,p)-\delta_{ts}}}\right),\\
\notag C_2&=\sqrt{2}\sigma(\Phi)\left(\frac{\phi(t,p)}{\phi(t,p)-\delta_{ts}}\frac{(2-p)(1-(t-1)\mu)}{2-p-(t-1)\mu}
\sqrt{1+\delta_{ts}}+\phi(t,p)\sqrt{\frac{1-p}{\phi(t,p)-\delta_{ts}}}\right)+1.
\end{align}
\end{theorem}

\begin{remark}
In the case of $d_i=1,~i=1,2,\cdots,M$, (\ref{eq.12}) is the same as Theorem 2 in \cite{Chen and Wan 2019}.
\end{remark}

\begin{theorem}\label{the.2}
Let $y=\Phi x+e$ be noisy measurements of a signal $x\in\mathbb{R}^N$ with $y,~e\in\mathbb{R}^n$, $\Phi\in\mathbb{R}^{n\times N}~(n<N)$ and $\|\Phi^{\top}e\|_{\infty}\leq\rho$. Assume that $\mathcal{B}=\mathcal{B}^{DS}(\varepsilon)$ with $\rho +\sigma^2(\Phi)\|x_{-\max(s)}\|_2\leq\varepsilon$ in (\ref{eq.4}). If $\Phi$ satisfies the block RIP with $\delta_{ts}<\phi(t,p)$ for some $1<t\leq 2$, then the solution $\hat{x}^{DS}$ to (\ref{eq.4}) fulfills
\begin{align}\label{eq.13}
\|\hat{x}^{DS}-x\|_2\leq D_1(\varepsilon+\rho)+D_2\|x_{-\max(s)}\|_2,
\end{align}
where
\begin{align}\notag
D_1&=\frac{\sqrt{2ds}\phi(t,p)}{\phi(t,p)-\delta_{ts}}\bigg(\frac{(2-p)(1-(t-1)\mu)}{2-p-(t-1)\mu}
+(1+\sqrt{N-ds})(1-p)\phi(t,p)\bigg),\\
\notag D_2&=\frac{\sqrt{2ds}\phi(t,p)\sigma^2(\Phi)}{\phi(t,p)-\delta_{ts}}\bigg(\frac{(2-p)(1-(t-1)\mu)}{2-p-(t-1)\mu}
+(1+\sqrt{N-ds})(1-p)\phi(t,p)\bigg)+1.
\end{align}
\end{theorem}

\begin{remark}
In the case of $d_i=1,~i=1,2,\cdots,M$, we obtain the same results as Theorem 3 in \cite{Chen and Wan 2019}.
\end{remark}

\begin{corollary}\label{cor.1}
Under the same conditions as in Theorem \ref{the.1}, suppose that $e=0$ and $x$ is block $s$-sparse. Then $x$ can be accurately reconstructed via
\begin{align}\label{eq.11}
\hat{x}=\arg\min_{\tilde{x}\in\mathbb{R}^N}\|\tilde{x}\|^p_{2,p}~\mbox{s.t.}~y=\Phi \tilde{x}.
\end{align}
\end{corollary}

In the following, we will decide how many random Gaussian measurements are demanded for (\ref{eq.7}) to be fulfilled with overwhelming probability. In this sequel, let $(\Omega,\tau)$ be a probability measure space and $z$ be a random variable which follows one of the following probability distributions:
\begin{align}\label{eq.33}
z\sim \mathcal{N}(0,1/n),~z\sim\begin{cases}1/\sqrt{n},\quad \ \ &
  \mbox{w.p.}~1/2,\\
  -1/\sqrt{n},\quad \ \ &
  \mbox{w.p.}~1/2,
\end{cases}
\mbox{or}~z\sim\begin{cases}\sqrt{3/n},\quad \ \ &
  \mbox{w.p.}~1/6,\\
  0,\quad \ \ &
  \mbox{w.p.}~2/3,\\
  -\sqrt{3/n},\quad \ \ &
  \mbox{w.p.}~1/6.
\end{cases}
\end{align}
Given $n$ and $N$, the random matrices $\Phi$ can be produced by making choice of the elements $\Phi_{i,j}$ as independent copies of $z$. This generates the random matrices $\Phi$.

\begin{theorem}\label{the.4}
Let $\Phi$ be an $n\times N$ matrix with $n<N$ whose elements are i.i.d. random variables defined by (\ref{eq.33}). If $$n\geq\frac{ts\log\frac{N}{ds}}{\frac{\phi^2(t,p)}{16}-\frac{\phi^3(t,p)}{48}},$$ the next assertion holds with probability more than \\$1-2\exp\left\{ts\left(d\log\frac{12}{\phi(t,p)}+\log\frac{e}{t}+\log\frac{M}{s}\right)
-n\left(\frac{\phi^2(t,p)}{16}-\frac{\phi^3(t,p)}{48}\right)\right\}$: for any block $s$-sparse signal $x\in\mathbb{R}^N$ over $\mathcal{I}=\{d_1=d,d_2=d,\cdots,d_M=d\}$ with $Md=N$, $x$ is the unique solution to (\ref{eq.11}) when the matrix $\Phi$ satisfies $\delta_{ts}<\phi(t,p)$.

\end{theorem}

\section{Numerical experiments}
\label{sec.5}

In this section, we carry out a few numerical simulations to hold out the application of our theoretical results. We could transform the constrained optimization problem (\ref{eq.4}) into an alternative unconstrained form below:
\begin{align}\label{eq.39}
\min_{\tilde{x}\in\mathbb{R}^N}\lambda\|\tilde{x}\|^p_{2,p}+\frac{1}{2}\|y-\Phi \tilde{x}\|^2_2.
\end{align}
Solving the problem (\ref{eq.39}), we adopt the standard Alternating Direction Method of Multipliers (ADMM)\cite{Wen F et al 2017a}\cite{Wen F et al 2017b}\cite{Lu CY et al 2018}. Utilizing a auxiliary variable $v\in\mathbb{R}^N$, we can rewrite the formulation (\ref{eq.39}) as
\begin{align}\label{eq.40}
\min_{\tilde{x},~v\in\mathbb{R}^N}\lambda\|v\|^p_{2,p}+\frac{1}{2}\|y-\Phi \tilde{x}\|^2_2~\mbox{s.t.}~\tilde{x}-v=0.
\end{align}
The augmented Lagrangian function of the above problem is
\begin{align}\label{eq.41}
\mathcal{L}_{\gamma}(\tilde{x},v,z)=\lambda\|v\|^p_{2,p}+\frac{1}{2}\|y-\Phi \tilde{x}\|^2_2+\left<z,\tilde{x}-v\right>+\frac{\gamma}{2}\|\tilde{x}-v\|^2_2,
\end{align}
where $z\in\mathbb{R}^n$ is a Lagrangian multiplier, and $\gamma>0$ is a penalty parameter. Then, ADMM composes of the below three steps:
\begin{align}
\label{eq.42}&\tilde{x}^{k+1}=\arg\min\frac{1}{2}\|y-\Phi \tilde{x}\|^2_2+\frac{\gamma}{2}\|z^k+\tilde{x}-v^k\|^2_2\\
\label{eq.43}&v^{k+1}=\arg\min\lambda\|v\|^p_{2,p}+\frac{\gamma}{2}\|z^k+\tilde{x}^{k+1}-v\|^2_2\\
\label{eq.44}&z^{k+1}=z^k+\tilde{x}^{k+1}-v^{k+1}.
\end{align}
The solution of problem (\ref{eq.42}) is explicitly provided by
\begin{align}\label{eq.45}
\tilde{x}^{k+1}=(\Phi^{\top}\Phi+\gamma I_n)^{-1}(\Phi^{\top}y-\gamma(z^k-v^k)).
\end{align}
To use the existing conclusions on the proximity operator of $l_p$-norm ($0<p<1$), by employing the inequality: $\|x\|_2<\|x\|_p$ for $x\in\mathbb{R}^n$ and $0<p<1$, the optimization problem (\ref{eq.43}) can be converted into
\begin{align}\label{eq.46}
v^{k+1}=\arg\min\sum^M_{i=1}\lambda\|v[i]\|^p_p+\frac{\gamma}{2}\|z^k[i]+\tilde{x}^{k+1}[i]-v[i]\|^2_2.
\end{align}

In our experiments, without loss of generality, we think over the block-sparse signal $x$ with even block size, i.e., $d_i=d$, $i=1,\cdots,M$, and take the signal length $N=1024$. For each experiment, first of all, we randomly produce block-sparse signal $x$ with the amplitude of each nonzero entry generated according to the Gaussian distribution. We use an $n\times N$ orthogonal Gaussian random matrix as the measurement matrix $\Phi$. We set the number of random measurement $m=128$ unless otherwise specified. With $x$ and $\Phi$, we generate the linear measurement $y$ by $y=\Phi x+e$, where $e$ is the Gaussian noise vector. Each given experimental result is an
average over 100 independent trails.

In Fig. \ref{fig.1}a, we produce the signals by making choice of $64$ blocks uniformly at random with $n=64,~N=128$, i.e., the block size $d=2$. The relative error of recovery $\|x-x^*\|_2/\|x\|_2$ is plotted versus the regularization parameter $\lambda$ for the different values of $p$, i.e., $p=0.2,~0.4,~0.6,~0.8,~1$. The $\lambda$ ranges from $10^{-8}$ to $10^{-2}$. From the figure, the parameter $\lambda=10^{-4}$ is a proper choice. Fig. \ref{fig.1}b presents experimental results regarding the performance of the non-block algorithm and the block algorithm with $p=0.4$. Two curves of relative error are provided via mixed $l_2/l_p$ minimization and orthogonal greedy algorithm (OGA) \cite{Petukhov 2006}. Fig. \ref{fig.1}b reveals the signal construction is quite significant in signal recovery.

\begin{figure}[h]
\begin{center}
\subfigure[]{\includegraphics[width=0.40\textwidth]{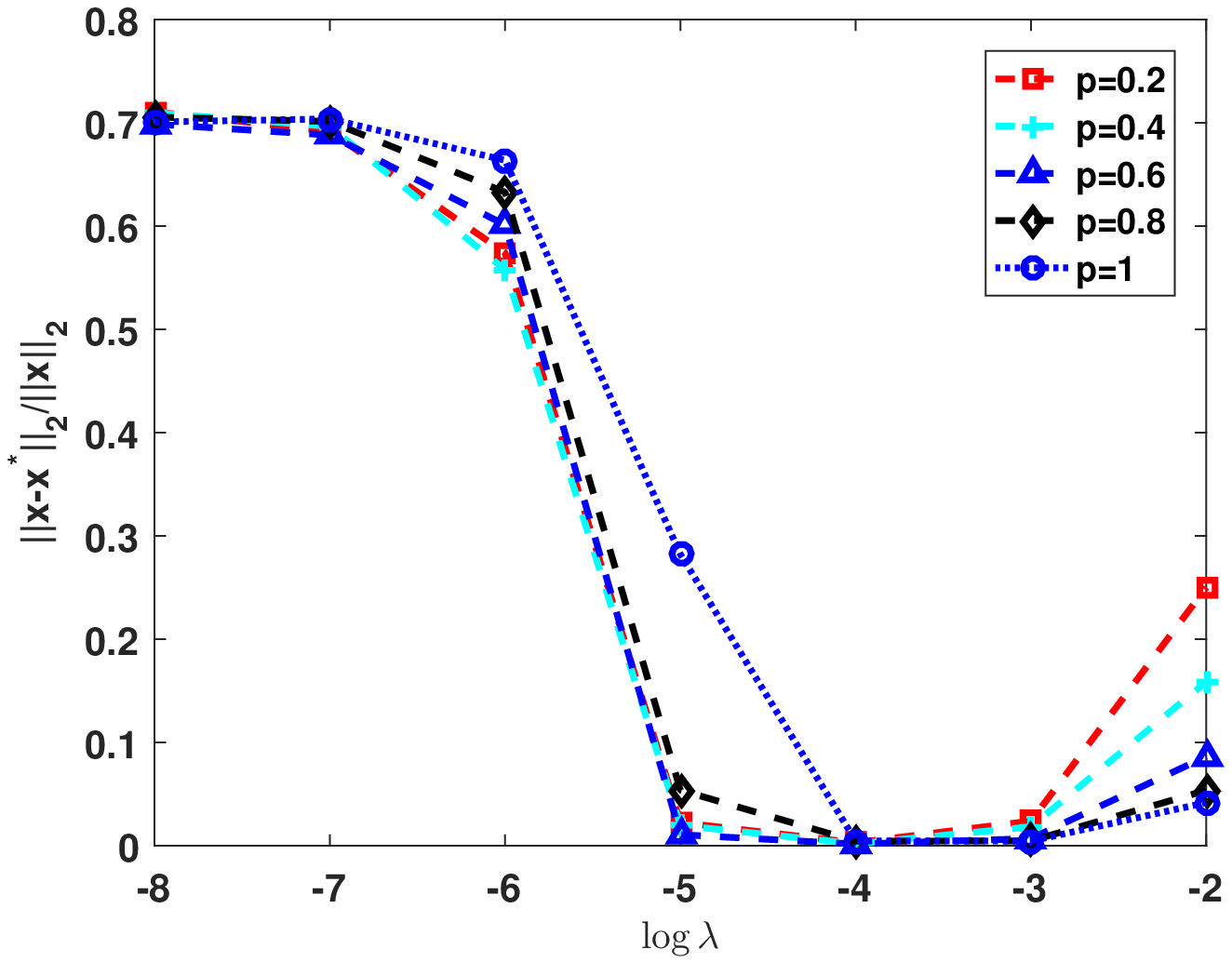}}
\subfigure[]{\includegraphics[width=0.40\textwidth]{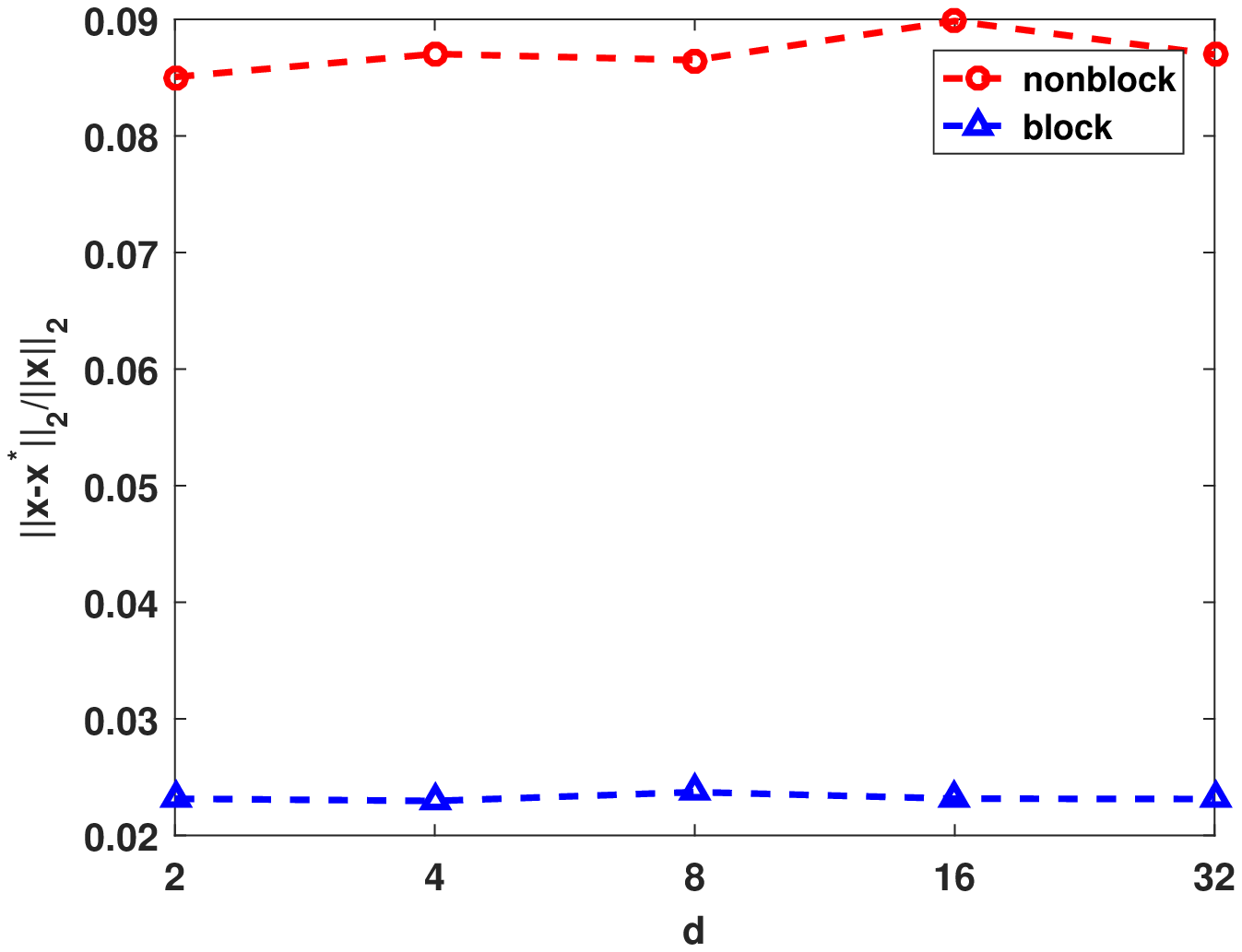}}
\caption{(a) Recovery performance of mixed $l_2/l_p$ minimization versus $\lambda$ for block size $d=2$, (b) Recovery performance of mixed $l_2/l_p$ minimization with $p=0.4$ and OGA and the number of nonzero entries $k=64$}\label{fig.1}
\end{center}
\vspace*{-14pt}
\end{figure}

Signal-to-noise ratio (SNR, SNR$=20\log_{10}(\|x\|_2/\|x-x^*\|_2)$) versus the values of $p$ and the nonzero entries $k$, the results are respectively given in Fig. \ref{fig.2}a and b. In Fig. \ref{fig.2}a, the values of $p$ vary from $0.01$ to $1$, and in Fig. \ref{fig.2}b, the number of nonzero entries $k$ ranges from $8$ to $48$. Figs. \ref{fig.2}a and b evidences that mixed $l_2/l_p$ minimization performs better than that of standard $l_p$ minimization. Figs. \ref{fig.2}a and b provide the relationship between the relative error and the number of measurements $n$ in different block sizes $d=1,~2,~4,~8$ and the values of $p=0.2,~0.4,~0.6,~0.8,~1$.

\begin{figure}[h]
\begin{center}
\subfigure[]{\includegraphics[width=0.40\textwidth]{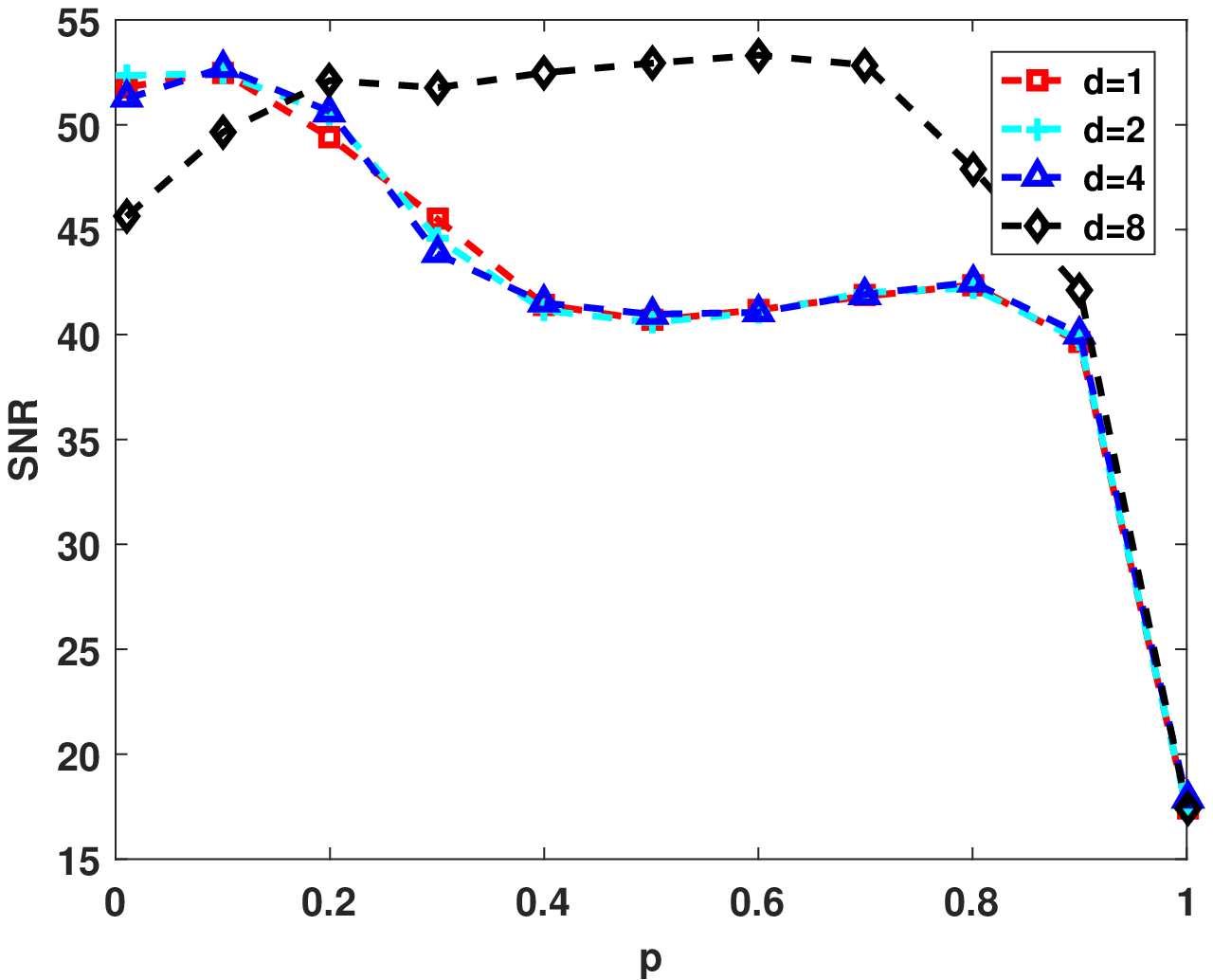}}
\hspace{0.5cm}
\subfigure[]{\includegraphics[width=0.40\textwidth]{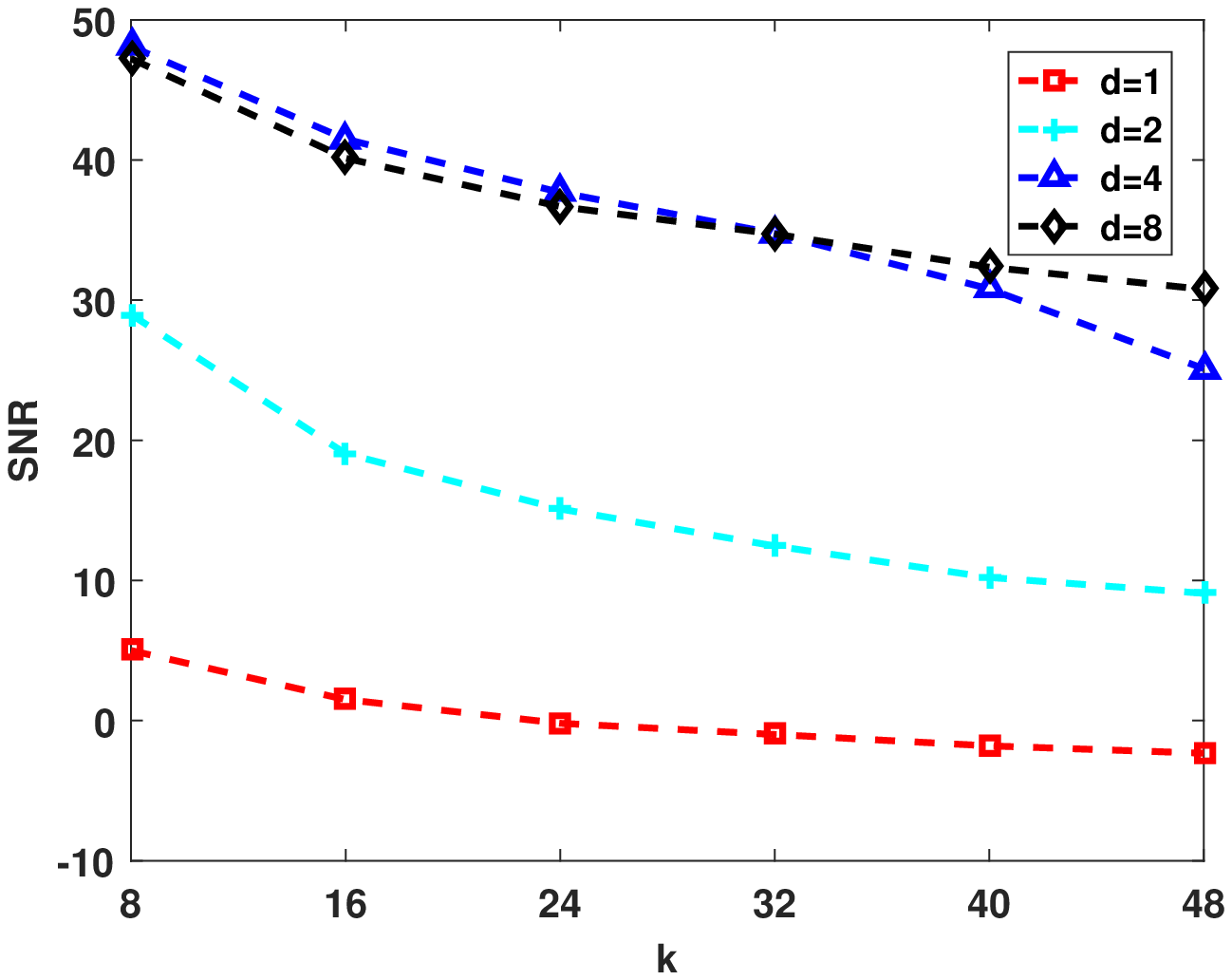}}
\caption{SNR versus the values of $p$ and the number of nonzero entries $k$ in (a) and (b) respectively (a) For $k=8$, (b) For $p=0.4$.}\label{fig.2}
\end{center}
\vspace*{-14pt}
\end{figure}

\begin{figure}[h]
\begin{center}
\subfigure[]{\includegraphics[width=0.40\textwidth]{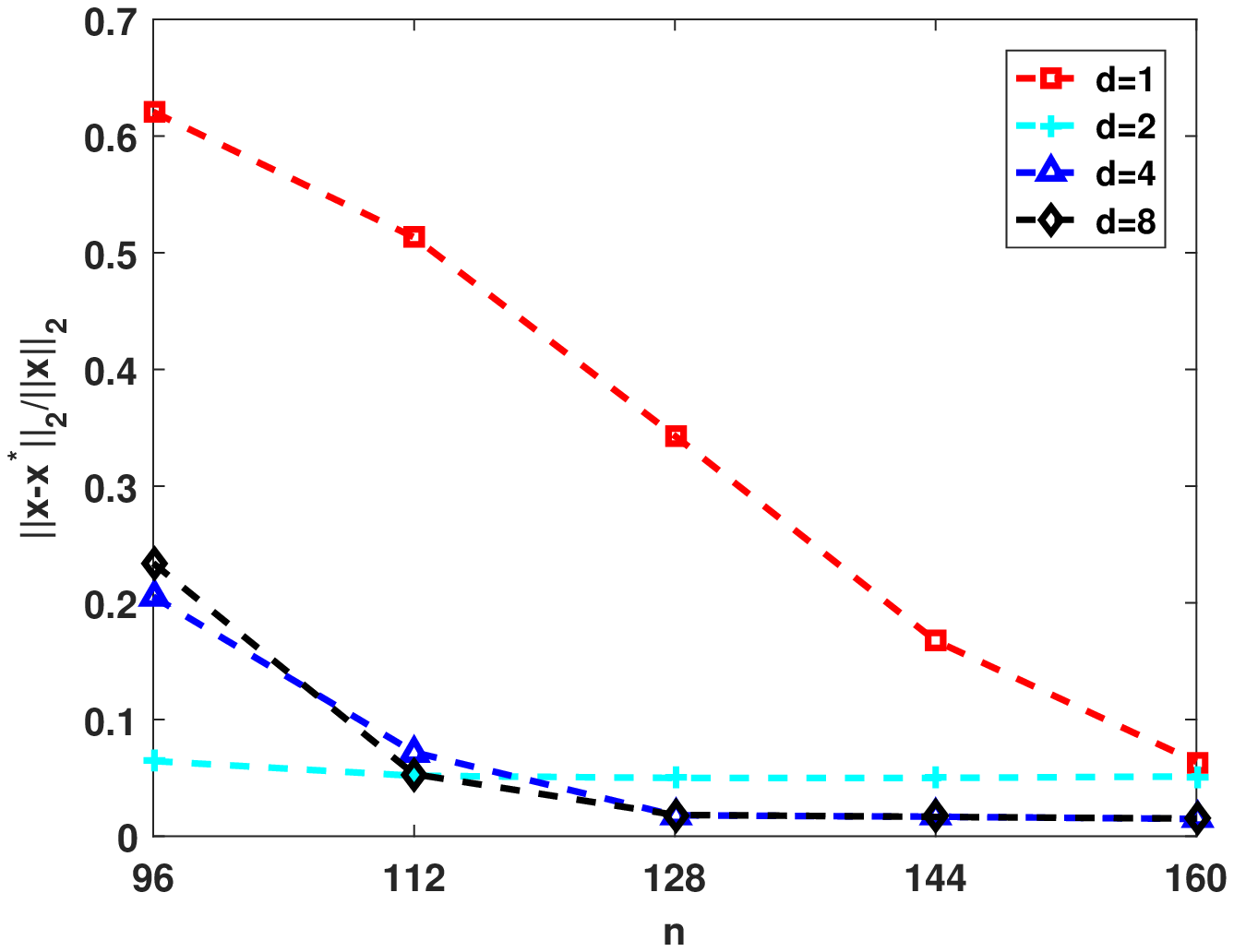}}
\hspace{0.5cm}
\subfigure[]{\includegraphics[width=0.40\textwidth]{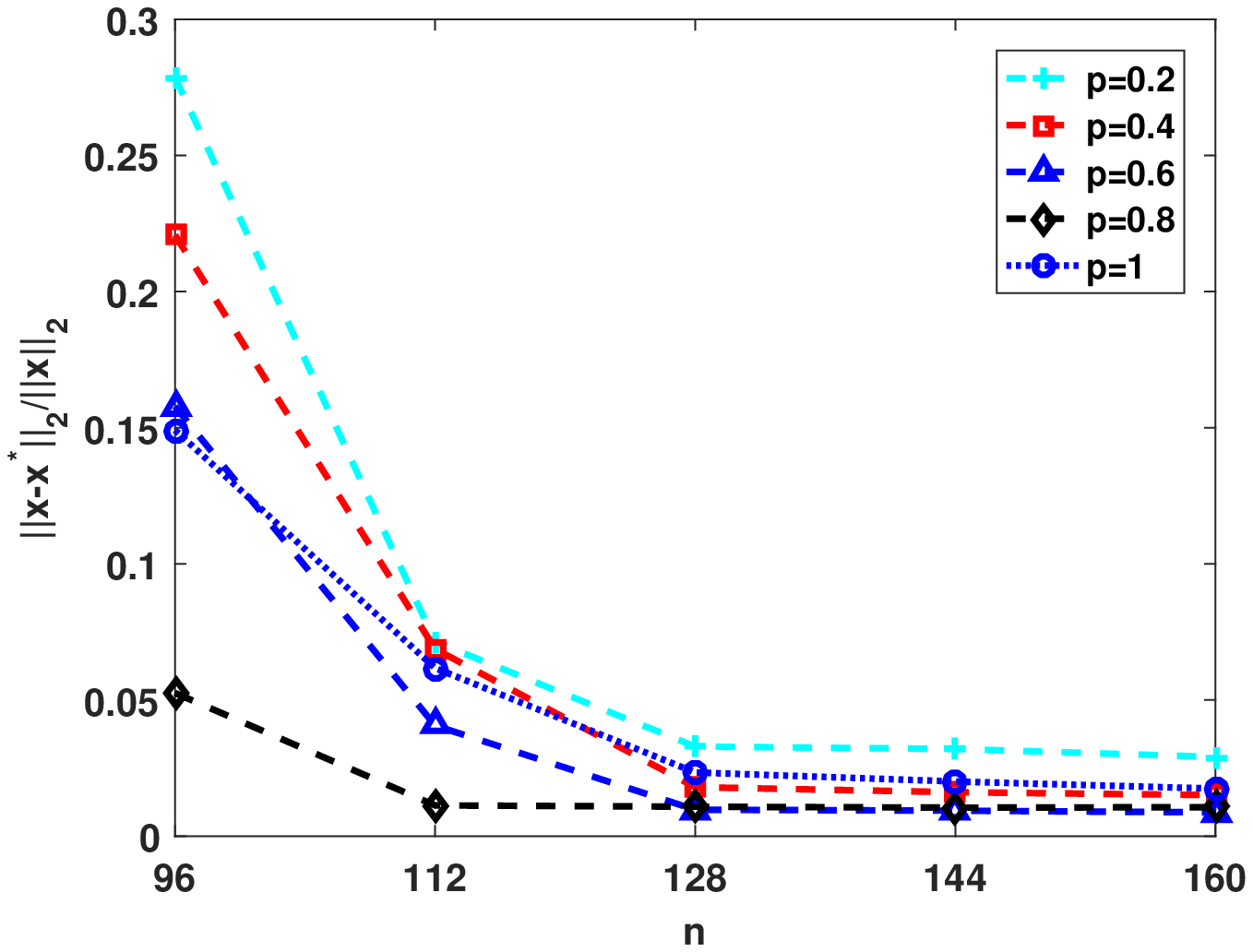}}
\caption{Relative error versus the number of measurements $n$ for (a) $p=0.4$, $k=16$ and (b) $d=2$, $k=16$}\label{fig.3}
\end{center}
\vspace*{-14pt}
\end{figure}

Eventually, we compare the performance of Group-Lp for $p=0.4$ with other typical algorithms consisting of Block-OMP algorithm\cite{Eldar and Kuppinger 2010}, Block-SL0 algorithm \cite{Ghalehjegh et al 2010} for $l_2/l_0$ solver and Block-ADM algorithm \cite{Deng W et al 2013}. We exploit SNR to weigh the algorithm efficiency. In Fig. \ref{fig.4}a, we select signals whose block size $d=2,~4,~8,~16,~32$ with the number of nonzero entries $k=32$ as the test signals, and in Fig. \ref{fig.4}b, the block size $d=2$. One can see that, overall, the performance of Group-Lp ($p=0.4$) is much better than that of other three algorithms.

\begin{figure}[h]
\begin{center}
\subfigure[]{\includegraphics[width=0.40\textwidth]{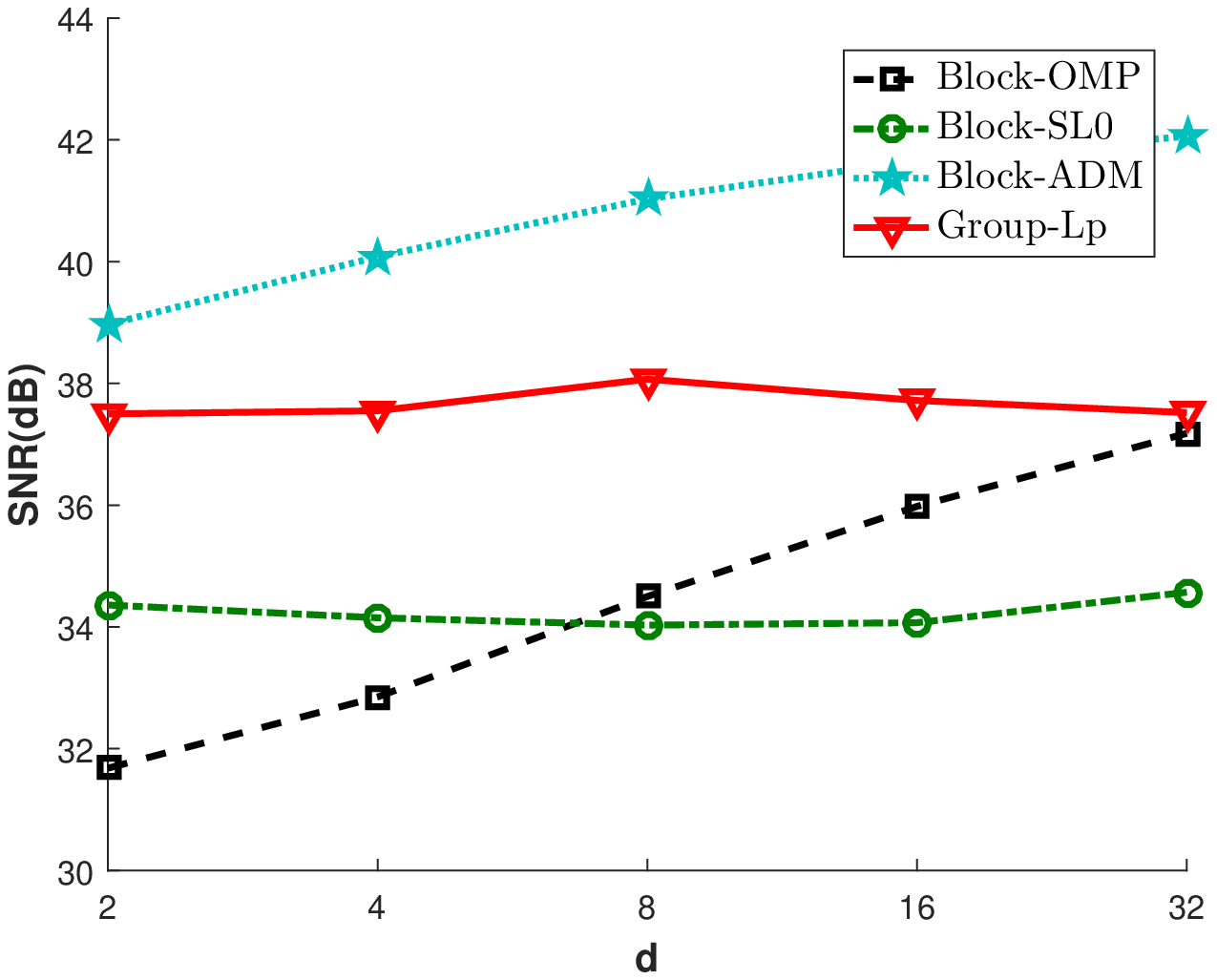}}
\hspace{0.5cm}
\subfigure[]{\includegraphics[width=0.40\textwidth]{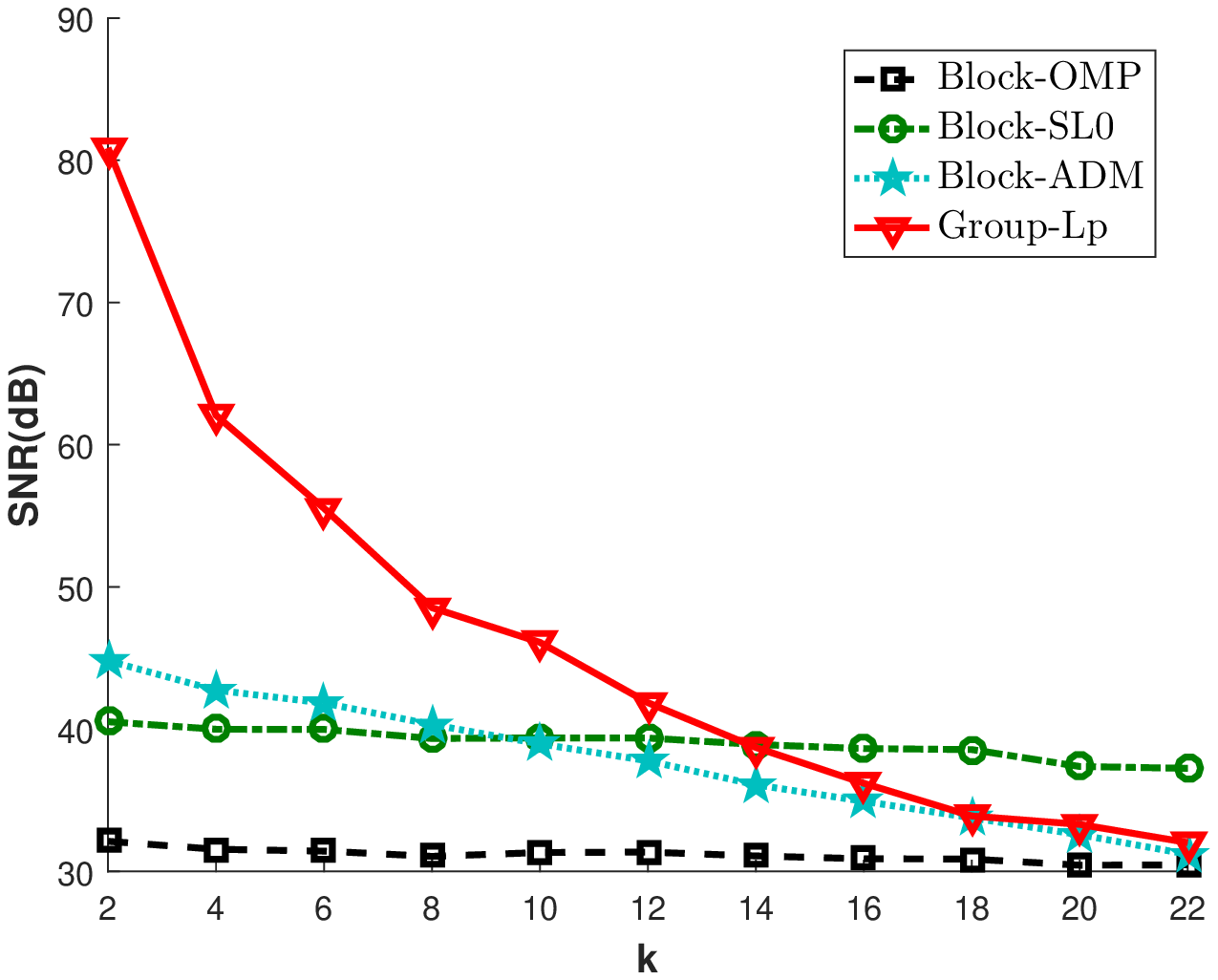}}
\caption{Comparison recovery performance with respect to SNR (a) Number of nonzero coefficients $k=32$ and (b) block size $d=2$}\label{fig.4}
\end{center}
\vspace*{-14pt}
\end{figure}

\section{The proofs of Main Results}
\label{sec.4}

\noindent \textbf{Proof of Theorem \ref{the.1}.}
First of all, suppose that $ts$ is an integer. Recollect that $h=\hat{x}^{l_2}-x_{\max(s)}$, where $\hat{x}^{l_2}$ is the solution to (\ref{eq.4}) with $\mathcal{B}=\mathcal{B}^{l_2}(\varepsilon)$. Then, by employing the condition of Theorem \ref{the.1}, we have
\begin{align}\notag
\|y-\Phi x_{\max(s)}\|_2\leq\|y-\Phi x\|_2+\|\Phi(x-x_{\max(s)})\|_2\leq\rho+\sigma(\Phi)\|x_{-\max(s)}\|_2\leq\varepsilon,
\end{align}
that is, $y-\Phi x_{\max(s)}\in\mathcal{B}^{l_2}(\varepsilon).$

Denote $\alpha^p=\|h_{\max(s)}\|^p_{2,p}/s$. Then, by utilizing Lemma \ref{lem.3}, we have
\begin{align}\label{eq.14}
\notag&\|h_{-\max(s)}\|_{2,\infty}\leq\alpha\leq\frac{\alpha}{(t-1)^{1/p}},~\mbox{and}\\
&\|h_{-\max(s)}\|^p_{2,p}\leq s(t-1)\left(\frac{\alpha}{(t-1)^{1/p}}\right)^p.
\end{align}
Through applying Lemma \ref{lem.1} and combining with (\ref{eq.14}), we have $h_{-\max(s)}=\sum^N_{i=1}\lambda_iu_i$, where $u_i$ is block $(t-1)s$-sparse, $\sum^N_{i=1}\lambda_i=1$ with $\lambda_i\in[0,1]$, and
\begin{align}\label{eq.15}
\sum_i\lambda_i\|u_i\|^2_{2,2}\leq\frac{\alpha^p}{t-1}\|h_{-\max(s)}\|^{2-p}_{2,2-p}.
\end{align}
Hence,
\begin{align}\label{eq.16}
\notag\sum_i\lambda_i\|u_i\|^2_{2,2}&\overset{\text{(a)}}{\leq}\frac{\alpha^p}{t-1}\left(\|h_{-\max(s)}\|^2_{2,2}\right)^{\frac{2(1-p)}{2-p}}
\left(\|h_{-\max(s)}\|^p_{2,p}\right)^{\frac{p}{2-p}}\\
\notag&\leq\frac{\alpha^p}{t-1}\left(\|h_{-\max(s)}\|^2_{2,2}\right)^{\frac{2(1-p)}{2-p}}(s\alpha^p)^{\frac{p}{2-p}}\\
\notag&\overset{\text{(b)}}{\leq}\frac{1}{t-1}\left(\|h_{-\max(s)}\|^2_{2,2}\right)^{\frac{2(1-p)}{2-p}}
\left(\|h_{\max(s)}\|^2_{2,2}\right)^{\frac{p}{2-p}}\\
&\overset{\text{(c)}}{\leq}\frac{1}{t-1}\left(\|h_{-\max(s)}\|^2_{2}\right)^{\frac{2(1-p)}{2-p}}
\left(\|h_{\max(s)}\|^2_{2}\right)^{\frac{p}{2-p}},
\end{align}
where (a) follows from H$\ddot{o}$lder's inequality, (b) is due to the fact that $\|x\|_p\leq\|x\|_q\leq n^{1/q-1/p}\|x\|_p$, $x\in\mathbb{R}^n$ and given $0<q<p\leq\infty$, and (c) is from the fact that $\|x\|^2_{2,2}=\sum^M_{i=1}\|x[i]\|^2_2=\sum^N_{i=1}x_i^2=\|x\|^2_2$, $x\in\mathbb{R}^N$.

From Cauchy-Schwartz inequality and the definition of block RIP, we get
\begin{align}\label{eq.17}
\left<\Phi h_{\max(s)},\Phi h\right>\leq\|\Phi h_{\max(s)}\|_2\|\Phi h\|_2\leq\sqrt{1+\delta_{ts}}\|h_{\max(s)}\|_2\|\Phi h\|_2.
\end{align}
Since
\begin{align}\label{eq.18}
\notag\|\Phi h\|_2&\leq\|\Phi(\hat{x}^{l_2}-x_{\max(s)})\|_2\leq\|\Phi\hat{x}^{l_2}-y\|_2
+\|\Phi x_{\max(s)}-y\|_2\\
&\leq\varepsilon+\rho+\sigma(\Phi)\|x_{-\max(s)}\|_2,
\end{align}
therefore
\begin{align}\label{eq.19}
\left<\Phi h_{\max(s)},\Phi h\right>\leq\sqrt{1+\delta_{ts}}(\varepsilon+\rho+\sigma(\Phi)\|x_{-\max(s)}\|_2)\|h_{\max(s)}\|_2.
\end{align}
Take $\beta_i=h_{\max(s)}+(t-1)\mu u_i$. Then,
\begin{align}\label{eq.20}
\sum^N_{j=1}\lambda_j\beta_j-\frac{p}{2}\beta_i-(t-1)\mu h=\left[1-(t-1)\mu-\frac{p}{2}\right]h_{\max(s)}-\frac{p}{2}(t-1)\mu u_i.
\end{align}
Furthermore, both $\sum^N_{j=1}\lambda_j\beta_j-\frac{p}{2}\beta_i-(t-1)\mu h$ and $\beta_i-\beta_j=(t-1)\mu(u_i-u_j)$ are block $ts$-sparse, because $h_{\max(s)}$ is block $s$-sparse, and $u_i$ is block $(t-1)s$-sparse.

Observe that the identity below \cite{Zhang and Li 2019}
\begin{align}\label{eq.21}
\notag\sum_i&\lambda_i\left\|\Phi\left(\sum_j\lambda_j\beta_j-\frac{p}{2}\beta_i\right)\right\|^2_2
+\frac{1-p}{2}\sum_{i,j}\lambda_i\lambda_j\|\Phi(\beta_i-\beta_j)\|^2_2\\
&=\left(1-\frac{p}{2}\right)^2\sum_i\lambda_i\|\Phi\beta_i\|^2_2,
\end{align}
where $\sum_i\lambda_i=1$.

First of all, we determine the left hand side (LHS) of (\ref{eq.21}). Putting (\ref{eq.20}) into LHS of (\ref{eq.21}) and combining with (\ref{eq.19}) and the concept of block RIP, we get
\begin{align}\label{eq.22}
\notag LHS=&\sum_i\lambda_i\|\Phi[(1-(t-1)\mu-\frac{p}{2})h_{\max(s)}-\frac{p}{2}(t-1)\mu u_i+(t-1)\mu h]\|^2_2\\
\notag&+\frac{1-p}{2}\sum_{i,j}\lambda_i\lambda_j(t-1)^2\mu^2\|\Phi(u_i-u_j)\|^2_2\\
\notag=&\sum_i\lambda_i\|\Phi[(1-(t-1)\mu-\frac{p}{2})h_{\max(s)}-\frac{p}{2}(t-1)\mu u_i]\|^2_2\\
\notag&+2(1-\frac{p}{2})(1-(t-1)\mu)(t-1)\mu<\Phi h_{\max(s)},\Phi h>+(1-p)(t-1)^2\mu^2\|\Phi h\|^2_2\\
\notag&+\frac{1-p}{2}(t-1)^2\mu^2\sum_{i,j}\lambda_i\lambda_j(t-1)^2\mu^2\|\Phi(u_i-u_j)\|^2_2\\
\notag\leq&(1+\delta_{ts})[\sum_i\lambda_i\|(1-(t-1)\mu-\frac{p}{2})h_{\max(s)}-\frac{p}{2}(t-1)\mu u_i\|^2_2\\
\notag&+\frac{1-p}{2}(t-1)^2\mu^2\sum_{i,j}\lambda_i\lambda_j(t-1)^2\mu^2\|u_i-u_j\|^2_2]\\
\notag&+2(1-\frac{p}{2})(1-(t-1)\mu)(t-1)\mu\sqrt{1+\delta_{ts}}
(\varepsilon+\rho+\sigma(\Phi)\|x_{-\max(s)}\|_2)\|h_{\max(s)}\|_2\\
\notag&+(1-p)(t-1)^2\mu^2(\varepsilon+\rho+\sigma(\Phi)\|x_{-\max(s)}\|_2)^2\\
\notag=&(1+\delta_{ts})[1-(t-1)\mu-\frac{p}{2}]^2\|h_{\max(s)}\|^2_2
+(1-\frac{p}{2})^2(t-1)^2\mu^2(1+\delta_{ts})\sum_i\lambda_i\|u_i\|^2_2\\
\notag&-(1-p)(t-1)^2\mu^2(1+\delta_{ts})\|h_{-\max(s)}\|^2_2\\
\notag&+2(1-\frac{p}{2})(1-(t-1)\mu)(t-1)\mu\sqrt{1+\delta_{ts}}
(\varepsilon+\rho+\sigma(\Phi)\|x_{-\max(s)}\|_2)\|h_{\max(s)}\|_2\\
&+(1-p)(t-1)^2\mu^2(\varepsilon+\rho+\sigma(\Phi)\|x_{-\max(s)}\|_2)^2.
\end{align}
Then again, by exploiting the representation of $\beta_i$ and the notion of block RIP, we get
\begin{align}\label{eq.23}
\notag RHS&=(1-\frac{p}{2})^2\sum_i\lambda_i\|\Phi\beta_i\|^2_2\\
\notag &=(1-\frac{p}{2})^2\sum_i\lambda_i\|\Phi(h_{\max(s)}+(t-1)\mu u_i)\|^2_2\\
\notag&\geq(1-\frac{p}{2})^2(1-\delta_{ts})\sum_i\lambda_i\|h_{\max(s)}+(t-1)\mu u_i\|^2_2\\
&=(1-\frac{p}{2})^2(1-\delta_{ts})\|h_{\max(s)}\|^2_2
+(1-\frac{p}{2})^2(t-1)^2\mu^2(1-\delta_{ts})\sum_i\lambda_i\|u_i\|^2_2.
\end{align}
A combination of (\ref{eq.22}) and (\ref{eq.23}), we obtain
\begin{align}\label{eq.24}
\notag&\{(1+\delta_{ts})[1-(t-1)\mu-\frac{p}{2}]^2-(1-\frac{p}{2})^2(1-\delta_{ts})\}\|h_{\max(s)}\|^2_2\\
\notag&+2(1-\frac{p}{2})^2(t-1)^2\mu^2\delta_{ts}\sum_i\lambda_i\|u_i\|^2_2
-(1-p)(t-1)^2\mu^2(1+\delta_{ts})\|h_{-\max(s)}\|^2_2\\
\notag&+2(1-\frac{p}{2})(1-(t-1)\mu)(t-1)\mu\sqrt{1+\delta_{ts}}
(\varepsilon+\rho+\sigma(\Phi)\|x_{-\max(s)}\|_2)\|h_{\max(s)}\|_2\\
&+(1-p)(t-1)^2\mu^2(\varepsilon+\rho+\sigma(\Phi)\|x_{-\max(s)}\|_2)^2\geq 0.
\end{align}
Substituting (\ref{eq.16}) into (\ref{eq.24}), we get
\begin{align}\label{eq.25}
\notag&\{(1+\delta_{ts})[1-(t-1)\mu-\frac{p}{2}]^2-(1-\frac{p}{2})^2(1-\delta_{ts})\}\|h_{\max(s)}\|^2_2\\
\notag&+2(1-\frac{p}{2})^2(t-1)\mu^2\delta_{ts}\left(\|h_{-\max(s)}\|^2_{2}\right)^{\frac{2(1-p)}{2-p}}
\left(\|h_{\max(s)}\|^2_{2}\right)^{\frac{p}{2-p}}\\
\notag&-(1-p)(t-1)^2\mu^2(1+\delta_{ts})\|h_{-\max(s)}\|^2_2\\
\notag&+2(1-\frac{p}{2})(1-(t-1)\mu)(t-1)\mu\sqrt{1+\delta_{ts}}
(\varepsilon+\rho+\sigma(\Phi)\|x_{-\max(s)}\|_2)\|h_{\max(s)}\|_2\\
&+(1-p)(t-1)^2\mu^2(\varepsilon+\rho+\sigma(\Phi)\|x_{-\max(s)}\|_2)^2\geq 0.
\end{align}
Regarding the LHS as the function of $\|h_{-\max(s)}\|^2_2$, we consider its extremum problem. Then,
\begin{align}\label{eq.26}
\notag&\{(1+\delta_{ts})[1-(t-1)\mu-\frac{p}{2}]^2-(1-\frac{p}{2})^2(1-\delta_{ts})\}\|h_{\max(s)}\|^2_2\\
\notag&+\frac{p}{2}(t-1)^2\mu^2(1+\delta_{ts})
(\frac{(2-p)\delta_{ts}}{(t-1)(1+\delta_{ts})})^{\frac{2-p}{p}}\|h_{\max(s)}\|^2_2\\
\notag&+2(1-\frac{p}{2})(1-(t-1)\mu)(t-1)\mu\sqrt{1+\delta_{ts}}
(\varepsilon+\rho+\sigma(\Phi)\|x_{-\max(s)}\|_2)\|h_{\max(s)}\|_2\\
&+(1-p)(t-1)^2\mu^2(\varepsilon+\rho+\sigma(\Phi)\|x_{-\max(s)}\|_2)^2\geq 0.
\end{align}

Noting that (\ref{eq.7}) and (\ref{eq.8}), we get
\begin{align}\label{eq.27}
\notag&[2-p-(t-1)\mu]^2(\delta_{ts}-\frac{\mu}{\frac{2-p}{t-1}-\mu})\|h_{\max(s)}\|^2_2\\
\notag&+2(1-\frac{p}{2})(1-(t-1)\mu)(t-1)\mu\sqrt{1+\delta_{ts}}
(\varepsilon+\rho+\sigma(\Phi)\|x_{-\max(s)}\|_2)\|h_{\max(s)}\|_2\\
&+(1-p)(t-1)^2\mu^2(\varepsilon+\rho+\sigma(\Phi)\|x_{-\max(s)}\|_2)^2\geq 0.
\end{align}
The condition $\delta_{ts}<\phi(t,p)$ guarantees that the inequality (\ref{eq.27}) is a second-order inequality for $\|h_{\max(s)}\|_2$, and the quadratic coefficient is less than zero. Consequently, we have
\begin{align}\label{eq.28}
\notag\|h_{\max(s)}\|_2\leq&\frac{1}{2[2-p-(t-1)\mu]^2(\phi(t,p)-\delta_{ts})}\{2(1-\frac{p}{2})(1-(t-1)\mu)(t-1)\mu\sqrt{1+\delta_{ts}}\theta\\
\notag&+\{[2(1-\frac{p}{2})(1-(t-1)\mu)(t-1)\mu\sqrt{1+\delta_{ts}}\theta]^2\\
\notag&+4[2-p-(t-1)\mu]^2(\phi(t,p)-\delta_{ts})(1-p)(t-1)^2\mu^2\theta^2\}^{\frac{1}{2}}\}\\
\overset{\text{(a)}}{\leq}&\{\frac{\phi(t,p)}{\phi(t,p)-\delta_{ts}}\frac{(2-p)(1-(t-1)\mu)}{2-p-(t-1)\mu}
\sqrt{1+\delta_{ts}}+\phi(t,p)\sqrt{\frac{1-p}{\phi(t,p)-\delta_{ts}}}\}\theta,
\end{align}
where (a) follows from the fact that $(u+v)^{1/2}\leq u^{1/2}+v^{1/2}$ for $u,~v\geq 0$, and $\theta=\varepsilon+\rho+\sigma(\Phi)\|x_{-\max(s)}\|_2.$ Combining with Lemmas \ref{lem.2} and \ref{lem.3}, it follows that $\|h_{-\max(s)}\|_2\leq\|h_{\max(s)}\|_2$.

Accordingly, it is not difficult to check that
\begin{align}\label{eq.29}
\notag\|\hat{x}^{l_2}-x\|_2&\leq\|\hat{x}^{l_2}-x_{\max(s)}\|_2+\|x-x_{\max(s)}\|_2\\
\notag&\leq\sqrt{2}\|h_{\max(s)}\|_2+\|x_{-\max(s)}\|_2\\
&\leq\sqrt{2}\{\frac{\phi(t,p)}{\phi(t,p)-\delta_{ts}}\frac{(2-p)(1-(t-1)\mu)}{2-p-(t-1)\mu}
\sqrt{1+\delta_{ts}}+\phi(t,p)\sqrt{\frac{1-p}{\phi(t,p)-\delta_{ts}}}\}\theta+\|x_{-\max(s)}\|_2.
\end{align}

If $ts$ is not an integer, we represent $t'=\lceil ts \rceil/s$, then $t's$ is an integer and $t<t'$. Due to $\frac{\partial \phi(t,p)}{\partial t}>0$, the function $\phi(t,p)$ is growing with $t\in(1,2]$. Thus, we have $\delta_{t's}=\delta_{ts}<\phi(t,p)<\phi(t',p)$. Analogous to the above proof, we can prove the result by working on $\delta_{t's}$.

\qed

\noindent \textbf{Proof of Theorem \ref{the.2}.}
Similar to the proof of the former noisy situation, for the case of noise type $\mathcal{B}=\mathcal{B}^{DS}(\rho)$, define $h=\hat{x}^{DS}-x_{\max(s)}$. We can derive
\begin{align}\notag
\|\Phi^{\top}(y-\Phi x_{\max(s)})\|_{\infty}\leq\|\Phi^{\top}(y-\Phi x)\|_{\infty}+\|\Phi^{\top}(\Phi x-\Phi x_{\max(s)})\|_{\infty}\leq\rho+\sigma^2(\Phi)\|x_{-\max(s)}\|_2\leq\varepsilon,
\end{align}
which reveals that $y-\Phi x_{\max(s)}\in\mathcal{B}^{DS}(\varepsilon).$ From the proof of Theorem \ref{the.1}, we have $\|h_{-\max(s)}\|_2\leq\|h_{\max(s)}\|_2$. By using the inequalities $\|x\|_p\leq\|x\|_q\leq n^{1/q-1/p}\|x\|_p$, $x\in\mathbb{R}^n$ and given $0<q<p\leq\infty$, we obtain $\|h_{-\max(s)}\|_1\leq\sqrt{N-ds}\|h_{\max(s)}\|_1$. Hence,
\begin{align}\label{eq.30}
\|h\|_1\leq(1+\sqrt{N-ds})\|h_{\max(s)}\|_1\leq(1+\sqrt{N-ds})\sqrt{ds}\|h_{\max(s)}\|_2.
\end{align}
Then,
\begin{align}\label{eq.31}
\notag\|\Phi h\|^2_2=&<h,\Phi^{\top}\Phi h>\leq\|h\|_1\|\Phi^{\top}\Phi h\|_{\infty}\\
\notag\leq&\|h\|_1(\|\Phi^{\top}(\Phi \hat{x}^{DS}-y)\|_{\infty}+\|\Phi^{\top}(\Phi x_{\max(s)}-y)\|_{\infty})\\
\leq&(1+\sqrt{N-ds})\sqrt{ds}\|h_{\max(s)}\|_2(\varepsilon+\rho+\sigma^2(\Phi)\|x_{-\max(s)}\|_2)
\end{align}
and
\begin{align}\label{eq.32}
\notag\left<\Phi h_{\max(s)},\Phi h\right>\leq&\left<h_{\max(s)},\Phi^{\top}\Phi h\right>
\leq\|h\|_1\|\Phi^{\top}\Phi h\|_{\infty}\\
\leq&\sqrt{ds}\|h_{\max(s)}\|_2(\varepsilon+\rho+\sigma^2(\Phi)\|x_{-\max(s)}\|_2).
\end{align}
Then, it follows that the below second-order inequality for $\|h_{\max(s)}\|_2$
\begin{align}\notag
\notag&[2-p-(t-1)\mu]^2(\phi(t,p)-\delta_{ts})\|h_{\max(s)}\|^2_2\\
\notag&-2(1-\frac{p}{2})(1-(t-1)\mu)(t-1)\mu\sqrt{ds}\|h_{\max(s)}\|_2
(\varepsilon+\rho+\sigma^2(\Phi)\|x_{-\max(s)}\|_2)\\
&-(1+\sqrt{N-ds})\sqrt{ds}(1-p)(t-1)^2\mu^2\|h_{\max(s)}\|_2(\varepsilon+\rho+\sigma^2(\Phi)\|x_{-\max(s)}\|_2)\leq 0.
\end{align}
Therefore,
\begin{align}
\notag\|h_{\max(s)}\|_2\leq&\frac{\phi(t,p)}{\phi(t,p)-\delta_{ts}}\{\frac{(2-p)(1-(t-1)\mu)\sqrt{ds}}{2-p-(t-1)\mu}
+(1+\sqrt{N-ds})\phi(t,p)(1-p)\sqrt{ds}\}\\
\notag&\times(\varepsilon+\rho+\sigma^2(\Phi)\|x_{-\max(s)}\|_2).
\end{align}
The remaining proof is similar with the case of noise type $\mathcal{B}^{l_2}$, we omit it here.

\qed

\noindent \textbf{Proof of Theorem \ref{the.4}.} Similar to the proof of Theorem 5.2 \cite{Baraniuk et al 2008}, from Lemma \ref{lem.4} and the union bound, for fixed $\delta\in(0,1)$, the measurement matrix $\Phi$ satisfies bock RIP (\ref{eq.3}) over $\mathcal{I}=\{d_1=d,d_2=d,\cdots,d_M=d\}$ with probability $\geq 1-2(\frac{12}{\delta})^{sd}(^M_s)e ^{-c_0(\delta/2)n}$, where $c_0(\delta/2)=\delta^2/16-\delta^3/48$ and $N=Md$. Therefore, we have
\begin{align}\label{eq.37}
\mathbb{P}(\delta_s<\delta)\geq1-2\left(\frac{12}{\delta}\right)^{sd}(^M_s)e ^{-c_0\left(\frac{\delta}{2}\right)n}
\end{align}
Corollary \ref{cor.1} shows that in the case of free-noise, the guarantee to exactly recover block $s$-sparse signals is $\delta_{ts}<\phi(t,p)~(1<t\leq 2)$. Hence, for $t\in(1,2]$, $\delta_{ts}<\phi(t,p)$ with probability
\begin{align}\label{eq.38}
\notag\mathbb{P}(\delta_{ts}<\phi(t,p))\geq&1-2\left(\frac{12}{\phi(t,p)}\right)^{tsd}(^M_{ts})e ^{-n\left(\frac{\phi^2(t,p)}{16}-\frac{\phi^3(t,p)}{48}\right)}\\
&\overset{\text{(a)}}{\leq}1-2e^{ts\left(d\log\frac{12}{\phi(t,p)}+\log\frac{e}{t}+\log\frac{M}{s}\right)
-n\left(\frac{\phi^2(t,p)}{16}-\frac{\phi^3(t,p)}{48}\right)},
\end{align}
where (a) follows from the inequality $(^u_v)\leq(eu/v)^v$ for integers $u>v>0$. When $M/s\to\infty$, to ensure that $\delta_{ts}<\phi(t,p)$ with overwhelming probability, the number of measurements must satisfy $n\geq ts\log\frac{N}{ds}/(\frac{\phi^2(t,p)}{16}-\frac{\phi^3(t,p)}{48})$.

\qed

\section{Conclusion}
\label{sec.6}

In recent years, the research of non-convex block-sparse compressed sensing has become a hot topic. This paper mainly discusses non-convex block-sparse compressed sensing by employing block RIP. We establish a sufficient condition that guarantee the stable and robust signal reconstruction via mix $l_2/l_p$ minimization method. Meanwhile, we present the upper bound estimation of recovery error. In addition, we give the number of samples needed to satisfy the sufficient conditions with high probability. Besides, we conduct a series of numerical experiments to show the verifiability of our results, and generally speaking, compared with other representative algorithms, the performance of Group-Lp algorithm is much better.



\begin{thebibliography}{100} \small

\bibitem{Majumdar and Ward 2010}
A. Majumdar , R. Ward ,Compressed sensing of color images, Signal Process. 90 (2010) 3122-3127.

\bibitem{Parvaresh et al 2008}
F. Parvaresh , H. Vikalo , S. Misra , B. Hassibi , Recovering sparse signals using sparse measurement matrices in compressed DNA microarrays, IEEE J. Sel. Top- ics Signal Process. 2 (3) (2008) 275-285 .

\bibitem{Cotter and Rao 2002}
S. Cotter , B. Rao , Sparse channel estimation via matching pursuit with applica- tion to equalization, IEEE Trans. Commun. 50 (3) (2002) 374-377 .

\bibitem{Huang J et al 2002}
J. Huang , X. Huang , D. Metaxas , Learning with Dynamic Group Sparsity, in: IEEE 12th International Conference on Computer Vision, 2009, pp. 64-71 .

\bibitem{Eldar and Mishali 2009}
Y. Eldar , M. Mishali , Robust recovery of signals from a structured union of sub- spaces, IEEE Trans. Inf. Theory 55 (11) (2009) 5302-5316 .

\bibitem{Eldar et al 2010}
Y. Eldar , P. Kuppinger , H. Bolcskei , Block-sparse signals: uncertainty relations and efficient recovery, IEEE Trans. Signal Process. 58 (6) (2010) 3042-3054 .

\bibitem{Lin and Li 2013}
J. Lin , S. Li , Block sparse recovery via mixed $l_2/l_1$ minimization, Acta Mathe- matica Sinica 29 (7) (2013) 1401-1412 .

\bibitem{Li and Chen 2019}
Li Y, Chen W. The high order block RIP condition for signal recovery[J]. Journal of Computational Mathematics, 2019, 37(1): 61-75.

\bibitem{Huang JW et al 2019}
Huang J, Wang J, Wang W, et al. Sharp sufficient condition of block signal recovery via $l_2/l_1$-minimisation[J]. IET Signal Processing, 2019, 13(5): 495-505.

\bibitem{Chartrand and Staneva 2008}
R. Chartrand , V. Staneva , Restricted isometry properties and nonconvex compressive sensing, Inverse Probl. 24 (2008) 1-14 .

\bibitem{Shen and Li 2012}
Y. Shen , S. Li , Restricted p -isometry property and its application for nonconvex compressive sensing, Adv. Comput. Math. 37 (2012) 441-452 .

\bibitem{Lai MJ et al 2013}
M. Lai , Y. Xu , W. Yin , Improved iteratively reweighted least squares for unconstrained smoothed l q minimization, SIAM J. Numer. Anal. 51 (2) (2013) 927-957 .

\bibitem{Majumdar and Ward 2010}
A. Majumdar , R. Ward ,Compressed sensing of color images, Signal Process. 90 (2010) 3122-3127 .

\bibitem{Wang Y et al 2013}
Y. Wang , J. Wang , Z. Xu , On recovery of block-sparse signals via mixed $l_2/l_q~(0<q\leq1)$ norm minimization, EURASIP J. Adv. Signal Process. 76 (2013) 1-17 .

\bibitem{Wang Y et al 2014}
Y. Wang , J. Wang , Z. Xu , Restricted $p$-isometry properties of nonconvex block-sparse compressed sensing, Signal Process. 104 (2014) 188-196 .

\bibitem{Wen JM et al 2019}
Wen J, Zhou Z, Liu Z, et al. Sharp sufficient conditions for stable recovery of block sparse signals by block orthogonal matching pursuit[J]. Applied and Computational Harmonic Analysis, 2019, 47(3): 948-974.

\bibitem{Gao Y et al 2017}
Gao Y, Peng J, Yue S. Stability and robustness of the $l_2/l_q$-minimization for block sparse recovery[J]. Signal Processing, 2017, 137: 287-297.

\bibitem{Wang JJ et al 2019}
Wang J, Huang J, Zhang F, et al. Group sparse recovery in impulsive noise via alternating direction method of multipliers[J]. Applied and Computational Harmonic Analysis, 2019.

\bibitem{Ge and Chen 2018}
Ge H, Chen W. Recovery of signals by a weighted $l_2/l_1$ minimization under arbitrary prior support information[J]. Signal Processing, 2018, 148: 288-302.

\bibitem{Li and Wen 2019}
 Haifeng Li and Jinming Wen. A New Analysis for Support Recovery With Block Orthogonal Matching Pursuit, IEEE Signal Processing Letters, 26(2019), 247-251.

\bibitem{Wang WD et al 2017}
Wang W, Wang J, Zhang Z. Block-sparse signal recovery via $l_2/l_{1-2}$ minimisation method[J]. IET Signal Processing, 2017, 12(4): 422-430.

\bibitem{Zhang and Li 2019}
Zhang R, Li S. Optimal RIP bounds for sparse signals recovery via $l_p$ minimization[J]. Applied and Computational Harmonic Analysis, 47(2019)566-584.

\bibitem{Cai and Zhang 2013}
T.T. Cai, A. Zhang, Sharp RIP bound for sparse signal and low-rank matrix recovery, Appl.
Comput. Harmon. Anal., 35 (2013), 74-93.

\bibitem{Chen and Wan 2019}
Chen B, Wan A. General RIP bounds of $\delta_{tk}$ for sparse signals recovery by $l_p$ minimization[J]. Neurocomputing, 2019, 363: 306-312.

\bibitem{Baraniuk et al 2008}
Baraniuk R, Davenport M, DeVore R, et al. A simple proof of the restricted isometry property for random matrices[J]. Constructive Approximation, 2008, 28(3): 253-263.

\bibitem{Lu CY et al 2018}
Lu C, Feng J, Lin Z, et al. Exact low tubal rank tensor recovery from Gaussian measurements[C]. Proceedings of the 27th International Joint Conference on Artificial Intelligence. AAAI Press, 2018: 2504-2510.

\bibitem{Wen F et al 2017a}
Wen F, Liu P, Liu Y, et al. Robust Sparse Recovery in Impulsive Noise via $\ell _p $-$\ell _1 $ Optimization[J]. IEEE Transactions on Signal Processing, 2017, 65(1): 105-118.

\bibitem{Wen F et al 2017b}
Wen F, Pei L, Yang Y, et al. Efficient and robust recovery of sparse signal and image using generalized nonconvex regularization[J]. IEEE Transactions on Computational Imaging, 2017, 3(4): 566-579.

\bibitem{Petukhov 2006}
Petukhov A. Fast implementation of orthogonal greedy algorithm for tight wavelet frames[J]. Signal processing, 2006, 86(3): 471-479.

\bibitem{Eldar and Kuppinger 2010}
Eldar Y C, Kuppinger P, Bolcskei H. Block-sparse signals: Uncertainty relations and efficient recovery[J]. IEEE Transactions on Signal Processing, 2010, 58(6): 3042-3054.

\bibitem{Ghalehjegh et al 2010}
Ghalehjegh S H, Babaie-Zadeh M, Jutten C. Fast block-sparse decomposition based on SL0[C]. International Conference on Latent Variable Analysis and Signal Separation. Springer, Berlin, Heidelberg, 2010: 426-433.

\bibitem{Deng W et al 2013}
Deng W, Yin W, Zhang Y. Group sparse optimization by alternating direction method[C]. Wavelets and Sparsity XV. International Society for Optics and Photonics, 2013, 8858: 88580R.



\end{thebibliography}
\end{document}